\documentclass[12pt]{article}
\usepackage[utf8]{inputenc}
\usepackage[usenames,dvipsnames,svgnames,table]{xcolor}
\usepackage[margin=1in]{geometry}
\usepackage{graphicx}
\usepackage{subcaption}
\usepackage{tcolorbox}
\usepackage{mathtools}
\usepackage{times}
\usepackage{sp}
\usepackage[numbers,sort&compress]{natbib}
\usepackage{thm-restate}

\usepackage{pgf,tikz,pgfplots}
\pgfplotsset{compat=1.15}
\usepackage{mathrsfs}
\usetikzlibrary{arrows}

\usepackage[colorlinks=true, citecolor=orange, linkcolor=black]{hyperref}

\newcommand{\bTheta}{\boldsymbol{\Theta}}

\title{{\bf Weakest Bidder Types and New Core-Selecting Combinatorial Auctions}}
\author{Siddharth Prasad\and Maria-Florina Balcan\and Tuomas Sandholm}
\date{}

\begin{document}

\maketitle

\begin{abstract}
  Core-selecting combinatorial auctions are popular auction designs that constrain prices to eliminate the incentive for any group of bidders---with the seller---to renegotiate for a better deal. They help overcome the low-revenue issues of classical combinatorial auctions. We introduce a new class of core-selecting combinatorial auctions that leverage bidder information available to the auction designer. We model such information through constraints on the joint {\em type space} of the bidders---these are constraints on bidders' private valuations that are known to hold by the auction designer before bids are elicited. First, we show that type space information can overcome the well-known impossibility of incentive-compatible core-selecting combinatorial auctions. We present a revised and generalized version of that impossibility result that depends on how much information is conveyed by the type spaces. We then devise a new family of core-selecting combinatorial auctions and show that they minimize the sum of bidders' incentives to deviate from truthful bidding. We develop new constraint generation techniques---and build upon existing quadratic programming techniques---to compute core prices, and conduct experiments to evaluate the incentive, revenue, fairness, and computational merits of our new auctions. Our new core-selecting auctions directly improve upon existing designs that have been used in many high-stakes auctions around the world. We envision that they will be a useful addition to any auction designer's toolkit.
\end{abstract}

\section{Introduction}

The design of \textit{combinatorial auctions (CAs)} is a complex task that requires careful engineering along several axes to best serve the application at hand. Just some of these axes are: taming cognitive and communication costs of eliciting and understanding bidders' inherently combinatorial valuations, tractable computation and optimization of economically efficient outcomes that allocate resources to those that value them the most, and determining prices that simplify bidders' incentives while generating acceptable revenues for the seller. These complexities are most evident in fielded applications of CAs including sourcing~\citep{Sandholm13:Very,Hohner03:Combinatorial,Sandholm06:Changing}, spectrum allocation~\citep{cramton2013spectrum,leyton2017economics}, treasury auctions~\citep{klemperer2010product}, and others.

The focus of the present paper is on better pricing rules for CAs. The classical{\em~\citet{Vickrey61:Counterspeculation}-\citet{Clarke71:Multipart}-\citet{Groves73:Incentives} (VCG) mechanism} is an economically efficient CA that is {\em incentive compatible (IC)}---a property of great practical importance since it levels the playing ground for bidders by making it worthless to strategize about their individual bids. But, the VCG auction has two major complementary issues (among others~\citep{Ausubel06:Lovely}) that prevent it from being practically viable: low revenue and prices that are not in the {\em core}. The latter means that some bidders might end up paying so little for their winnings that others who offered more for those same items would take issue. {\em Core-selecting CAs} fix this issue with prices that ensure no coalition of bidders plus the seller would want to renegotiate for a better deal, but these give up on incentive compatibility. Therefore, most core-selecting CA designs use core prices that minimize bidders' incentives to deviate from truthful bidding.

Core-selecting CAs have been used to auction licenses for wireless spectrum by a number of countries' governments including Australia, Canada, Denmark, Ireland, Mexico, the Netherlands, Portugal, Switzerland, the United Kingdom, and others, generating many billions of dollars in revenue~\citep{cramton2013spectrum,palacios2024combinatorial}. \citet{ausubel2017market} review some of the key design choices of the FCC incentive auction that was completed in the United States in 2017. They suggest that some instances of winners paying zero for certain packages despite losers bidding competitively~\citep{ausubel2023vcg} could have been avoided with a core-selecting payment rule instead of the VCG rule adopted by the FCC (though a core-selecting rule would have introduced other practical difficulties in other stages of the auction). While the most prominent real-world deployment of core-selecting CAs is probably spectrum auctions, their use has been proposed for other important applications such as electricity markets~\citep{karaca2019core}, advertisement markets~\citep{goetzendorff2015compact,niazadeh2022fast}, and auctions for wind farm development rights~\citep{ausubel2011auction}.

In this paper we introduce a new class of core-selecting CAs that improve upon prior  designs by taking advantage of bidder information available to the auction designer through constraints on the bidders' {\em type spaces}. Our starting point is the {\em weakest-type VCG (WT) auction}, which is a type-space-dependent improvement of VCG~\citep{krishna1998efficient,balcan2023bicriteria}. Our core-selecting CAs build upon the WT auction, and minimize the sum of bidders' incentives to deviate from truthful bidding. They generalize and improve upon the core-selecting CA designs that have been developed in the literature so far, some of which have been successfully used in spectrum auctions~\citep{Day07:Fair,day2010optimal,day2012quadratic,erdil2010new}.

\subsection{Our Contributions}

First, we show that information expressed by type spaces can overcome the following well-known impossibility result due to~\citet{Othman10:Envy} and~\citet{goeree2016impossibility}: under unrestricted type spaces either (i) VCG is not in the core in which case no IC core-selecting CA exists or (ii) VCG is the unique IC core-selecting CA. In general CAs where bidders' valuations exhibit complementarities (that is, the value of a bundle is more than the sum of its parts), VCG is typically not in the core. VCG is in the core only under strict conditions on bidder valuations that rule out complementarity (like buyer-submodularity or gross-substitutes~\citep{Ausubel02:Ascending}). We provide a revised and more general version of the impossibility result. Our result (Theorem~\ref{theorem:impossibility}) states that either (i) WT is not in the core in which case no IC core-selecting CA exists, (ii) WT is the unique IC core-selecting CA, or (iii) there are infinitely many IC core-selecting CAs including WT (and we characterize all such CAs). In particular, vanilla VCG has no bearing on the existence of IC core-selecting CAs (when type spaces are unrestricted VCG and WT are identical, so our result recovers the one by~\citet{Othman10:Envy} and~\citet{goeree2016impossibility} in that case).

Second, we devise a new family of type-space-dependent core-selecting CAs that minimize the sum of bidders' incentives to deviate from truthful bidding. Typical core-selecting CAs choose prices that lie on the minimum-revenue face---referred to as the {\em minimum-revenue core (MRC)}---of the core polytope~[\citealp{Parkes01:Achieving}, \citealp{Day07:Fair}, \citealp{day2010optimal}, \citealp{erdil2010new}, \citealp{day2012quadratic}].~\citet{day2010optimal} show that MRC points minimize bidders' total incentive to deviate from truthful bidding (and therefore minimize incentives to deviate in a Pareto sense as well). Our new design chooses core prices that minimize revenue subject to the additional constraint that they lie above WT. We generalize Day and Milgrom's result (which hinges on the assumption of unrestricted typespaces), and show that our revised version of the minimum-revenue core provides optimal incentives for bidders.

Third, we develop new constraint generation routines for computing WT prices. We compare two linear programming formulations of WT price computation: one is due to~\citet{balcan2023bicriteria} and the other is based on~\citet{Bikhchandani02:Package}. Both linear programs have an exponential number of constraints, so we develop constraint generation routines to solve them. In our experiments, the~\citet{balcan2023bicriteria} formulation leads to significantly smaller constraint-generation solve times and iterations. 
%The above two sentences constitute a weird phrasing. It makes it sound like our contribution is worse than what was already in teh literature. Reword. ???
On most instances, WT price computation via our constraint generation routine only adds a modest run-time overhead to the cost of winner determination.

Finally, we present proof-of-concept experiments that evaluate the incentive, revenue, and fairness properties of our new core-selecting CAs. We coin and implement three new core-selecting payment rules that select payments on our revised MRC. Our implementation uses the quadratic programming and core-constraint generation technique developed by~\citet{day2012quadratic}.

\subsection{Related Work}

\paragraph{Weakest types}

The notion of a weakest type consistent with an agent's type space originates from the seminal works of~\citet{Myerson83:Efficient} and~\citet{cramton1987dissolving} in the context of efficient trade. It was first presented in an auction context by~\citet{krishna1998efficient}, and later modified by~\citet{balcan2023bicriteria} to derive revenue guarantees that depend on measures of informativeness of the type space. The weakest-type auction has found applications in other mechanism design settings (like digital goods auctions) as well~\citep{lu2024competitive}.

\paragraph{Equilibrium bidding strategies in core-selecting CAs}

As core-selecting CAs are not generally incentive compatible, there is a sizable literature that studies bidding strategies and equilibrium outcomes in core-selecting CAs. Such work has generally been limited to very small CA instances with numbers of items and bidders both in the single digits. ~\citet{goeree2016impossibility} derive equilibrium strategies for the core-selecting CA of~\citet{day2012quadratic} and show that {\em revealed} core prices can be further away from the {\em true} core than VCG. ~\citet{ausubel2020core} are more optimistic and demonstrate the opposite phenomenon, providing more justification for the use of core-selecting CAs in practice. \citet{bichler2013core} run lab experiments to study bidding behavior and efficiency of the core-selecting combinatorial clock auction format.~\citet{ott2013incentives} study overbidding equilibria that can arise in core-selecting CAs.

%\citet{sano2011incentives} characterizes when VCG is in the core when bidders are single minded (though this rarely occurs).

\paragraph{Core-selecting CA design and computation}
\citet{erdil2010new} introduce the idea of using ``reference points'' other than VCG~\citep{day2012quadratic} to find closest MRC prices. \citet{bunz2022designing} perform a computational evaluation of different core-selecting payment rules that differ in their underlying reference point. Their focus is on computing equilibrium bidding strategies (using modern Bayes-Nash equilibrium solvers~\citep{bosshard2017computing}) to evaluate true efficiency, and therefore their evaluation is limited to very small auction instances. These works study the properties of different core points that lie on the same minimum-revenue core. We redefine the minimum-revenue core to depend on the type space information known to the auction designer.

\citet{bunz2015faster} provide improvements to the original core-constraint generation algorithms of~\citet{Day07:Fair} and~\citet{day2012quadratic}. ~\citet{niazadeh2022fast} develop non-exact algorithms that converge to core prices (though their experimental evaluation is in a not-fully-combinatorial advertising setting where winner determination is in $\mathsf{P}$, in contrast to the general CA setting where winner determination is $\mathsf{NP}$-complete). Generalizing their algorithms to take advantage of type space information is an interesting direction for future research. \citet{goel2015core,markakis2019core} devise incentive compatible CAs that approximate the core revenue. A drawback of this line of work is that it sacrifices efficiency, which is one of the main tenets that motivates the need for core-selecting CAs in the first place. \citet{goetzendorff2015compact} design new bidding languages for auctions with many items and respective techniques for core pricing; \citet{moor2016core} study core-selecting auctions when some items might no longer be available after the auction is run; \citet{Othman10:Envy} develop an iterative core-selecting CA that elicits bids over multiple rounds.

\paragraph{Core selection beyond CAs} Some work has studied the design of core-selecting mechanisms in markets beyond auctions. Examples include combinatorial exchanges~\citep{rostek2015core, bichler2017core}, reallocative mechanisms like the FCC incentive auction~\citep{rostek2023reallocative}, and markets with financially-constrained buyers~\citep{batziou2022core,bichler2022core}.

\section{Problem Formulation and Background on Core-Selecting CAs}

In a combinatorial auction (CA) there is a set $M=\{1,\ldots,m\}$ of indivisible items to be auctioned off to bidders $N = \{1,\ldots,n\}$ who can submit bids for distinct bundles (or packages) of items. Bidder $i$ reports to the auction designer her valuation $v_i:2^M\to\R_{\ge 0}$ that encodes the maximum value $v_i(S)$ she is willing to pay for every distinct bundle of goods $S\subseteq M$. Let $\vec{v} = (v_1,\ldots, v_n)$ denote the valuation profile of all bidders, and let $\vec{v}_{-i} = (v_1,\ldots, v_{i-1}, v_{i+1},\ldots, v_n)$ denote the profile of bids excluding bidder $i$. For $C\subseteq N$ let $\vec{v}_C = (v_j)_{j\in C}$ and let $\vec{v}_{-C} = (v_j)_{j\in N\setminus C}$. We assume bidders report their valuations in the XOR bidding language~[\citealp{Sandholm99:Algorithm}, \citealp{Nisan00:Bidding}], under which a bidder can only win at most one of the bundles she explicitly placed a nonzero bid for. For bidder $i$, let $B_i\subseteq 2^M$ be the set of bundles she bid on (assume for notational convenience that each bidder $i$ implicitly submits $v_i(\emptyset) = 0$). Let $\Gamma=\Gamma(B_1,\ldots, B_n)\subseteq B_1\times\cdots\times B_n$ denote the set of feasible allocations, that is, the set of partitions $S_1,\ldots, S_n$ of $M$ with $S_i\in B_i$ for each $i$ and $S_i\cap S_j=\emptyset$ for each $i, j$. We use boldface $\vec{S}=(S_1,\ldots,S_n)\in\Gamma(B_1,\ldots, B_n)$ to denote a feasible allocation.

%{\color{blue} Say we do not deal with communication issues and that bidders can express their full valuation. IN experiments of course bidders will bid on some tractable number of bundles}

Before bids/valuations are submitted, bidder $i$'s valuation $v_i$, also called her {\em type}, is her own private information. The auction designer might have some prior information about the bidders, and that is modeled by the {\em joint type space} of the bidders, denoted $\bTheta\subseteq\bigtimes_{i\in N}\R^{2^m}$. The auction designer knows that  $\vec{v}\in\bTheta$. Given $\vec{v}_{-i}$, let $\Theta_i(\vec{v}_{-i}) = \{\hat{v}_i : (\hat{v}_i, \vec{v}_{-i})\in \bTheta\}$ be the projected type space of bidder $i$. So, after seeing the revealed bids $\vec{v}_{-i}$ of all other bidders, the auction designer knows $v_i\in\Theta_i(\vec{v}_{-i})$. This model of type spaces begets a rich and expressive language of bidder information available to the auction designer---$\bTheta$ can represent any statement of the form ``the joint valuation profile $\vec{v}$ of all bidders satisfies property $P$'' (\citet{balcan2023bicriteria} provide concrete examples). The typical assumption in mechanism design is an unrestricted type space $\Theta = \bigtimes_{i\in N}\R^{2^m}_{\ge 0}$ (what is usually assumed is the existence of a known prior distribution over the type space). In contrast, we will be concerned with explicit representations of the auction designer's knowledge via the type space and how that influences both practical computation and the auction design itself.

\subsubsection*{Auction design desiderata} An auction is determined by its allocation rule and its payment rule. In this paper we are concerned with {\em efficient auctions}. An efficient auction selects the efficient (welfare-maximizing) allocation: $$\vec{S}^*=(S_1^*,\ldots, S_n^*) = \argmax_{\vec{S}\in \Gamma}\sum_{j\in N}v_j(S_j).$$ The {\em winner determination problem} of computing the efficient allocation is {\sf NP}-complete (by a reduction from weighted set packing), but solving its integer programming formulation is generally a routine task for modern integer programming solvers. Let $w(\vec{v}) = \max_{S\in\Gamma}\sum_{j\in N} v_j(S_j)$ denote the efficient welfare. An auction is {\em incentive compatible (IC)} if each bidder's utility (value minus payment) is weakly maximized by truthful bidding, independent of other bidders. An auction is {\em individually rational (IR)} if truthful bidders are always guaranteed non-negative utility, independent of other bidders.

%IC, IR, etc. IC and IR need only hold over $\bTheta$, but the misreporting space is unrestricted. Restrictions on the type space relax the IC and IR constraints, making the auction design space richer.

\subsubsection*{Vickrey-Clarke-Groves (VCG) Auction} The classical auction due to~\citet{Vickrey61:Counterspeculation},~\citet{Clarke71:Multipart}, and~\citet{Groves73:Incentives} (VCG) chooses the efficient allocation $\vec{S}^*$, and charges bidder $i$ a payment of $$p^{\texttt{VCG}}_i(\vec{v}) = w(0, \vec{v}_{-i}) - \sum_{j\neq i}v_j(S_j^*).$$ Let $\vec{p}^{\texttt{VCG}}=(p_1^{\texttt{VCG}},\ldots,p_n^{\texttt{VCG}})$ denote the vector of VCG payments. VCG is incentive compatible and individually rational. Generally, to implement the VCG auction one must solve winner determination $n+1$ times---once to compute $w(\vec{v})$ and the efficient allocation, and once per bidder to compute $w(0, \vec{v}_{-i})$ in the formula for $p_i$.

\subsubsection*{Weakest-Type VCG (WT) Auction} The weakest-type VCG (WT) auction~\citep{krishna1998efficient,balcan2023bicriteria} chooses the efficient allocation $\vec{S}^*$ achieving welfare $w(\vec{v})$ and charges bidder $i$ a payment of \begin{equation}\label{eq:wcvcg_payment}p^{\WT}_i(\vec{v}) = \min_{\widetilde{v}_i\in \Theta_i(\vec{v}_{-i})}w(\widetilde{v}_i, \vec{v}_{-i}) - \sum_{j\neq i}v_j(S_j^*).\end{equation}
%If $p_i^{\WT} > v_i(S_i^*)$, bidder $i$ receives no items and is charged nothing. Otherwise $i$ receives $S_i^*$ and is charged $p_i^{\WT}$.
In Equation~\eqref{eq:wcvcg_payment}, the bid vector $\widetilde{v}_i$ achieving the minimum is the {\em weakest type} in $\Theta_i(\vec{v}_{-i})$. Let $\vec{p}^{\WT} = (p_1^{\WT},\ldots, p_n^{\WT})$ denote the vector of WT payments. If $\Theta = \bigtimes_{i\in N}\R^{2^m}_{\ge 0}$, so $\Theta_i(\vec{v}_{-i}) = \R_{\ge 0}^{2^m}$ conveys no information about bidder $i$, the weakest competitor $\widetilde{v}_i$ is the bidder who bids zero on all packages, and $p_i^{\WT} = p_i^{\VCG}$. On the other extreme, if $\Theta = \{\vec{v}\}$, $p_i^{\WT} = v_i(S_i^*)$ and WT extracts the total social surplus as revenue. For $\bTheta$ closed and convex, WT is revenue maximizing among all efficient, IC, and IR auctions~\citep{krishna1998efficient,balcan2023bicriteria}. (The intuition behind this fact is that WT closes the slack in the IR constraint left open by vanilla VCG---the weakest type's IR constraint is binding under WT.) A weaker Bayes-IC and Bayes-IR version of the WT auction was first introduced by~\citet{krishna1998efficient}. The version we work with here is due to~\citet{balcan2023bicriteria}.

\subsubsection*{Core-Selecting CAs and the Minimum-Revenue Core} Let $W = \{i\in N : S_i^*\neq\emptyset\}$ be the set of winning bidders in the efficient allocation $\vec{S}^*$. A combinatorial auction is in the {\em core} if (i) it chooses the efficient allocation $\vec{S}^*$ and (ii) prices $\vec{p}$ lie in the {\em core polytope}, defined by {\em core constraints} for every group of winning bidders and IR constraints:
\begin{equation}\label{eq:core}
    \core(\vec{v}) = \left\{\vec{p}\in\R^W :
    \begin{aligned}
        & \sum_{i\in W\setminus C}p_i\ge w(\vec{0}, \vec{v}_C) - \sum_{j\in C}v_j(S_j^*)~\forall C\subseteq N, \\
        & v_i(S_i^*) - p_i\ge 0\;\;\forall i\in W
    \end{aligned}\right\}.
\end{equation} This formulation of the core gives rise to a direct interpretation of core prices as ``group VCG prices'': any set of winners must in aggregate pay the externality they impose on the other bidders (our formulation is not the typical formulation of the core, which is a notion originally from cooperative game theory, but is most convenient from an implementation/mathematical programming perspective as in~\citet{Day07:Fair,day2012quadratic,bunz2015faster}). When $W\setminus C = \{i\}$ is a singleton, the core constraint reads $p_i\ge p_i^{\texttt{VCG}}$.

The {\em minimum-revenue core (MRC)} is the set $\MRC=\argmin\{\norm{\vec{p}}_1: \vec{p}\in\core\}$ that consists of all core prices of minimal revenue.~\citet{Day07:Fair,day2010optimal} show that the MRC captures exactly the set of core prices that minimize the sum of bidders' incentives to deviate from truthful bidding. The MRC is not unique and there can be (infinitely) many MRC prices. Some core-selecting CAs that select unique MRC points that have been proposed are VCG nearest~\citep{day2012quadratic}, which finds the MRC point closest in Euclidean distance to VCG, and zero nearest~\citep{erdil2010new}, which finds the MRC point closest in Euclidean distance to the origin.

Since core-selecting CAs are in general not IC, a core-selecting CA only guarantees that prices are in the {\em revealed} core with respect to reported bids. But, from a regulatory viewpoint, the revealed core is nonetheless a useful solution concept since core constraints prevent any group of bidders from lodging a meaningful complaint based on their actual bids~\citep{bunz2022designing}.

\section{Impossibility of IC Core-Selecting CAs}\label{sec:impossibility}

We revisit the following dichotomy for core-selecting CAs when type spaces are unrestricted~\citep{goeree2016impossibility, Othman10:Envy}: either (i) VCG is not in the core which implies no IC core-selecting auction exists or (ii) VCG is in the core and is the unique IC core-selecting auction. That dichotomy relies on the assumption that $\Theta$ is unrestricted, that is, $\Theta = \R^{2^m}_{\ge 0}$. We revise and generalize that result to depend on bidders' type spaces. The proof relies on the revenue optimality of WT prices subject to efficiency, IC, and IR~\citep{balcan2023bicriteria,krishna1998efficient}.

\begin{theorem}\label{theorem:impossibility}
    Let $\bTheta$ be closed and convex.
    Let $\vec{v}$ be the vector of bidders' true valuations. If $\vec{p}^{\WT}(\vec{v})\notin\core(\vec{v})$, no incentive compatible core-selecting CA exists. Otherwise, let $\mathfrak{C}\subseteq 2^N$ be the set of core constraints that $\vec{p}^{\WT}$ satisfies with equality. Let $\mathfrak{C}' = \{C'\subseteq N : C'\cap C=\emptyset~\forall C\in\mathfrak{C}\}$ and for $C'\in\mathfrak{C}'$ let $$s(C') = \sum_{i\in W\setminus C'}p_i^{\WT} - w(\vec{0},\vec{v}_{C'}) + \sum_{j\in C'}v_j(S_j^*)$$ be the slack of the $C'$-core constraint. Then for any $C'\in\mathfrak{C}'$ all prices in the set $$\left\{\left(\vec{p}^{\WT}_{W\cap C'}-\vec{\varepsilon}, \vec{p}^{\WT}_{W\setminus C'}\right) : \norm{\vec{\varepsilon}}_1\le s(C'),\vec{\varepsilon}\in\R_{\ge 0}^{W\cap C'}\right\}$$ are in the core and are attainable via an incentive compatible CA.
\end{theorem}

\begin{proof}
    If $\vec{p}^{\WT}\notin\core(\vec{v})$, it must be the case that for any $\vec{p}\in\core(\vec{v})$ there exists $i$ such that $p_i > p_i^{\WT}$. This means no IC core-selecting CA can exist because $\vec{p}^{\WT}$ is bidder-wise payment optimal subject to efficiency, IC, and IR~\citep{balcan2023bicriteria}.

    If $\vec{p}^{\WT}\in\core(\vec{v})$, the price vector $(\vec{p}^{\WT}_{W\cap C'}-\vec{\varepsilon}, \vec{p}^{\WT}_{W\setminus C'})$ is also in the core for any $\vec{\varepsilon}$ with $\norm{\vec{\varepsilon}}_1\le s(C')$ by construction. We now argue that there exists an IC auction that yields these prices.
    Consider the efficient Groves mechanism that uses pivot terms $$h_i(\vec{v}_{-i}) = t_i\cdot w(0, \vec{v}_{-i}) + (1-t_i)\cdot \min_{\widetilde{v}_i\in\Theta_i(\vec{v}_{-i})}w(\widetilde{v}_i,\vec{v}_{-i})$$ where $t_i\in [0,1]$ is a parameter that does not depend on $i$'s revealed type $v_i$. Such a Groves mechanism is IC and, since it produces payments between VCG and WT, IR. By continuity, there exist parameters $\vec{t} = ((t_i)_{i\in W\cap C'}, \vec{0})$ so that the Groves mechanism produces prices $(\vec{p}^{\WT}_{W\cap C'}-\vec{\varepsilon}, \vec{p}^{\WT}_{W\setminus C'})$.
\end{proof}

Theorem~\ref{theorem:impossibility} implies that if WT is in the core, there is a potential continuum of IC core-selecting payment rules obtained by decreasing WT prices along non-binding faces of the core. In particular, the existence of IC core-selecting CAs does not depend on VCG prices but on WT prices. WT and VCG coincide when type spaces do not convey sufficient information about the additional welfare created by a bidder: $p^{\WT}_i = p^{\VCG}_i$ if and only if $\min_{\widetilde{v}_i\in\Theta_i(\vec{v}_{-i})}w(\widetilde{v}_i,\vec{v}_{-i}) = w(0, \vec{v}_{-i})$, which says that the information conveyed by $\Theta_i(\vec{v}_{-i})$ about bidder $i$ is so weak that it cannot even guarantee that $i$'s presence adds any nonzero welfare to the auction. In this case, Theorem~\ref{theorem:impossibility} recovers the impossibility result of~\citet{Othman10:Envy} and~\citet{goeree2016impossibility}.

\section{Our New Core-Selecting CAs and their Properties}\label{sec:incentives}

In this section we introduce our new class of core-selecting CAs based on weakest types, and prove that it provides bidders with optimal incentives (by minimizing the sum of bidders' incentives to deviate, therefore providing optimal incentives in a Pareto sense as well) among all core-selecting CAs. Our result generalizes the result of~\citet{day2010optimal} which was in the setting of unrestricted type spaces (our result recovers theirs in the unrestricted case).

In Section~\ref{sec:impossibility} we have shown above that if WT is not in the core, then all core-selecting CAs necessarily violate incentive compatibility. To measure the incentive violations of a core-selecting CA, we borrow the notion of an incentive profile from~\citet{day2010optimal}. The {\em utility profile} (resp., {\em deviation profile}) of an efficient CA with payment rule $\vec{p}(\vec{v})$ is given by $\{\mu_i^{\vec{p}}(\vec{v})\}_{i\in W}$ (resp. $\{\delta_i^{\vec{p}}(\vec{v})\}_{i\in W}$), where $$\mu_i^{\vec{p}}(\vec{v}) = \max_{\widehat{v}_i} \left(v_i(\hat{S}_i) - p_i(\widehat{v}_i,\vec{v}_{-i})\right)$$ is bidder $i$'s maximum obtainable utility from misreporting and $$\delta_i^{\vec{p}}(\vec{v}) = \max_{\widehat{v}_i}\left(v_i(\hat{S}_i) - p_i(\widehat{v}_i,\vec{v}_{-i})\right) - \left(v_i(S_i^*) - p_i(v_i,\vec{v}_{-i})\right) = \mu_i^{\vec{p}}(\vec{v}) - \left(v_i(S_i^*) - p_i(v_i,\vec{v}_{-i})\right)$$ is bidder $i$'s maximum utility gain over truthful bidding ($\hat{\vec{S}}$ denotes the efficient allocation under reported bid profile $(\hat{v}_i, \vec{v}_{-i})$). Our goal is to define core-selecting payment rules $\vec{p}$ that minimize the sum of bidders' incentives to deviate, which is precisely $\sum_i \delta_i^{\vec{p}}(\vec{v})$. The quantity $\delta_i^{\vec{p}}$ can be viewed as a form of ex-post regret for truthful bidding for bidder $i$.
Throughout this section, $\vec{v}$ denotes the true valuations of the bidders.

%{\color{blue}Required fact about $p^{\WT}_i$ being the smallest report one can make such that they continue to win their efficient bundle, and stuff is IR.}

\begin{comment}
\begin{definition}
    A core-selecting payment rule $\vec{p}$ {\em provides optimal incentives} at $\vec{v}$ if there is no other core-selecting $\vec{p}'$ such that $\varepsilon_i^{\vec{p}'}(\vec{v})\le \varepsilon_i^{\vec{p}}(\vec{v})$ for all $i$ and $\varepsilon_{i^*}^{\vec{p}'}(\vec{v}) < \varepsilon_{i^*}^{\vec{p}}(\vec{v})$ for some $i^*$.
\end{definition}
\end{comment}

The following lemma generalizes~\citet[Theorem 3.2]{Day07:Fair}; its proof is identical to theirs.

\begin{lemma}\label{lemma:max_utility}
    Let $\hat{\vec{p}}$ be any payment rule that implements the efficient allocation such that $\hat{p}_i\ge p_i^{\WT}$. Then, $\mu_i^{\hat{\vec{p}}}(\vec{v})\le v_i(S_i^*) - p_i^{\WT}(\vec{v})$ and $\delta_i^{\hat{\vec{p}}}(\vec{v})\le \hat{p}_i(\vec{v}) - p_i^{\WT}(\vec{v})$. That is, the maximum utility winner $i$ can obtain by misreporting under $\hat{\vec{p}}$ is no more than her utility under $\vec{p}^{\WT}$.
\end{lemma}

\begin{proof}
    Suppose for the sake of contradiction that there is a misreport $v_i'$ for bidder $i$ that gives her utility more than $v_i(S_i^*) - p_i^{\WT}(\vec{v})$, that is, $v_i (S_i') - \hat{p}_i(v_i',\vec{v}_{-i}) > v_i(S_i^*) - p_i^{\WT}(\vec{v})$ where $\vec{S}'$ is the efficient allocation for bid profile $(v_i',\vec{v}_{-i})$. Since $\hat{p}_i\ge p_i^{\WT}$, $v_i(S_i') - p_i^{\WT}(v_i', \vec{v}_{-i})\ge v_i(S_i') - \hat{p}_i (v_i', \vec{v}_{-i})$, which, combined with the above, yields $v_i(S_i') - p_i^{\WT}(v_i', \vec{v}_{-i}) > v_i(S_i^*) - p_i^{\WT}(\vec{v})$. Incentive compatibility of WT is violated, a contradiction.
\end{proof}

The following result is an adaptation of~\citet[Theorem 2]{day2010optimal}; the proof is similar to theirs.

\begin{theorem}\label{theorem:best_response}
    Let $\hat{\vec{p}}$ be any IR payment rule that implements the efficient allocation such that $\hat{p}_i \ge p_i^{\WT}$. Let $v_i'$ denote the misreport for winner $i$ defined by $v_i'(S_i^*) = p_i^{\WT}(\vec{v})$, $v_i'(S) = 0$ for all $S\neq S_i^*$. Then, $v_i'$ is a best response for $i$ that gives her utility equal to $v_i(S_i^*) - p_i^{\WT}(\vec{v})$. That is, under $\hat{\vec{p}}$, winner $i$ can always guarantee herself utility equal to what her utility would have been under $\vec{p}^{\WT}$.
\end{theorem}

\begin{proof}
    Reporting $v_i'$ does not change the efficient allocation since $p_i^{\WT}\ge p_i^{\VCG}$ (and $i$'s VCG price is her lowest possible misreport that preserves her winning bundle). So, the IR constraint for $\hat{\vec{p}}$ requires $v_i'(S_i^*) - \hat{p}_i(v_i', \vec{v}_{-i})\ge 0$. Expanding the left-hand side yields $v_i'(S_i^*) - \hat{p}_i(v_i', \vec{v}_{-i}) = p_i^{\WT}(\vec{v}) - \hat{p}_i(v_i', \vec{v}_{-i}) = v_i(S_i^*) - (w(v_i, \vec{v}_{-i}) - \min_{\widetilde{v}_i\in\Theta_i(\vec{v}_{-i})} w(\widetilde{v}_i,\vec{v}_{-i})) - \hat{p}_i(v_i', \vec{v}_{-i})$. So, $v_i(S_i^*) - \hat{p}_i(v_i', \vec{v}_{-i}) \ge w(v_i, \vec{v}_{-i}) - \min_{\widetilde{v}_i\in\Theta_i(\vec{v}_{-i})} w(\widetilde{v}_i,\vec{v}_{-i})$. The right-hand side is precisely $i$'s utility under $\vec{p}^{\WT}$. By Lemma~\ref{lemma:max_utility}, this constitutes a best response.
\end{proof}

Theorem~\ref{theorem:best_response} allows us to characterize the subset of points that minimize the sum of bidders' incentives to deviate of any upwards closed region. They are exactly the set of points of minimal revenue. Given a price vector $\hat{\vec{p}}\in\R^W$ and any closed region $\cA\subseteq\R^W$ , let $$\mathsf{MR}_{\cA}(\hat{\vec{p}}) = \argmin\left\{\norm{\vec{p}}_1 : \vec{p}\in\cA, \hat{\vec{p}}\le\vec{p}\le (v_i(S_i^*))_{i\in W}\right\}$$ be the set of IR price vectors in $\cA$ of minimal revenue that lie above $\hat{\vec{p}}$.

\begin{theorem}\label{theorem:general_uc}
    Let $\cA\subseteq\R^W$ be upwards closed. Then $$\mathsf{MR}_{\cA}(\vec{p}^{\WT})\subseteq\argmin\left\{\sum_{i\in W}\delta_i^{\vec{p}}(\vec{v}) : \vec{p}\in \cA\right\}.$$
\end{theorem}

\begin{proof}

    Consider the map on pricing rules $\vec{p}\mapsto\vec{p}'$ defined by $$p'_i(\vec{v}) = 
    \begin{cases}
        p_i(\vec{v}) & p_i(\vec{v})\ge p_i^{\WT}(\vec{v}) \\
        p_i^{\WT}(\vec{v}) & p_i(\vec{v}) < p_i^{\WT}(\vec{v}).
    \end{cases}$$

    This map satisfies the property that $\delta_i^{\vec{p}'}(\vec{v})\le\delta_i^{\vec{p}}(\vec{v})$ since if $p_i(\vec{v})\ge p_i^{\WT}(\vec{v})$, $p_i$ is unchanged, and otherwise the WT price is used for which $\delta_i^{\vec{p}^{\WT}}(\vec{v}) = 0$ due to incentive compatibility. So, for any price vector $\vec{p}\in\cA$ such that $\vec{p}\ngeq\vec{p}^{\WT}$, the described map produces $\vec{p}'$ such that $\vec{p}'\in\cA$ (since $\cA$ is upwards closed), $\vec{p}'\ge\vec{p}^{\WT}$, and $\vec{p}'$ has deviation profile no worse than $\vec{p}$. It therefore suffices to consider the subset of $\cA$ that lies above $\vec{p}^{\WT}$ to find prices in $\cA$ that minimize the sum of bidders' incentives to deviate. For $\vec{p}\ge\vec{p}^{\WT}$ we have $\delta_i^{\vec{p}}(\vec{v}) = p_i(\vec{v}) - p_i^{\WT}(\vec{v})$ by Theorem~\ref{theorem:best_response}. So minimizing $\sum_i \delta_i^{\vec{p}}$ is equivalent to minimizing $\sum_i p_i$, which completes the proof.
\end{proof}

Let $\MRC(\widehat{\vec{p}}) = \mathsf{MR}_{\core}(\widehat{\vec{p}})$ denote the minimum-revenue core above $\widehat{\vec{p}}$. Applying Theorem~\ref{theorem:general_uc} yields:

\begin{corollary}\label{corollary:core}
    $\MRC(\vec{p}^{\WT}(\vec{v}))\subseteq\argmin\left\{\sum_{i\in W}\delta_i^{\vec{p}}(\vec{v}):\vec{p}\in\core(\vec{v})\right\}$.
\end{corollary}

Any payment rule $\vec{p}\in\MRC(\vec{p}^{\WT}(\vec{v}))$ is therefore incentive optimal in a Pareto sense as well: there is no other core-selecting $\vec{p}'$ such that $\delta_i^{\vec{p}'}(\vec{v})\le\delta_i^{\vec{p}}(\vec{v})$ for all $i$ and $\delta_{i^*}^{\vec{p}'}(\vec{v})\le\delta_{i^*}^{\vec{p}}(\vec{v})$ for some $i^*$. Corollary~\ref{corollary:core} generalizes the results of~\citet{Day07:Fair,day2010optimal} since when $\Theta_i(\vec{v}_{-i})$ is unrestricted for each agent $i$, $\MRC(\vec{p}^{\WT}) = \MRC(\vec{p}^{\VCG})$ which is the (unrestricted) minimum-revenue core they consider.

Corollary~\ref{corollary:core} gives strong theoretical justification for payment rules that lie on $\MRC(\vec{p}^{\WT})$. We expand on specific rules in Section~\ref{sec:experiments}, but as one concrete example one of the rules we coin---{\em WT nearest}---selects the price vector in $\MRC(\vec{p}^{\WT})$ that minimizes Euclidean distance to $\vec{p}^{\WT}$. WT nearest is the most direct generalization of the VCG nearest rule proposed by~\citet{day2012quadratic} that has been successfully used in spectrum auctions. In order to implement rules like WT nearest, we need algorithms for computing $\vec{p}^{\WT}$. That is the topic of the next section (Section~\ref{sec:formulations}). We conclude this section with an example illustrating some of the key concepts introduced so far.

\begin{example}\label{example:core}
Consider the CA with three items $\{a, b, c\}$ and $10$ single-minded bidders who submit the following bids: $v_1(a) = 20$, $v_2(b) = 20$, $v_3(c) = 20$, $v_4(ab)=28$, $v_5(ac) = 26$, $v_6(bc)=23$, $v_7(a)=10$, $v_8(b)=10$, $v_9(c)=10$, $v_{10}(abc)=41$ (this a slight modification of an example from~\citet{day2012quadratic}). Bidders 1, 2, and 3 win in the efficient allocation and their VCG prices are $\vec{p}^{\VCG}=(10,10,10)$. Say $$\Theta_1 = \R_{\ge 0},\Theta_2 = \{v_2(b)\ge 17\},\Theta_3=\{v_3(c)\ge 15\},$$ so $\vec{p}^{\WT}=(10,17,15)$. The core constraints are given by $$\left\{p_1, p_2, p_3\ge 10, p_1+p_2\ge 28, p_1+p_3\ge 26, p_2+p_3\ge 23, p_1+p_2+p_3\ge 41\right\}.$$ The vanilla VCG-nearest point of~\citet{day2012quadratic} on $\MRC(\vec{p}^{\VCG})$ is $(14, 14, 13)$ and the WT-nearest point on $\MRC(\vec{p}^{\WT})$ is $(11, 17, 15)$. Figure~\ref{fig:core} is an illustration of this example.
\end{example}

\begin{figure}[t]
    \centering
    \includegraphics[width=0.6\linewidth, trim={5cm 0cm 7.5cm, 0cm},clip]{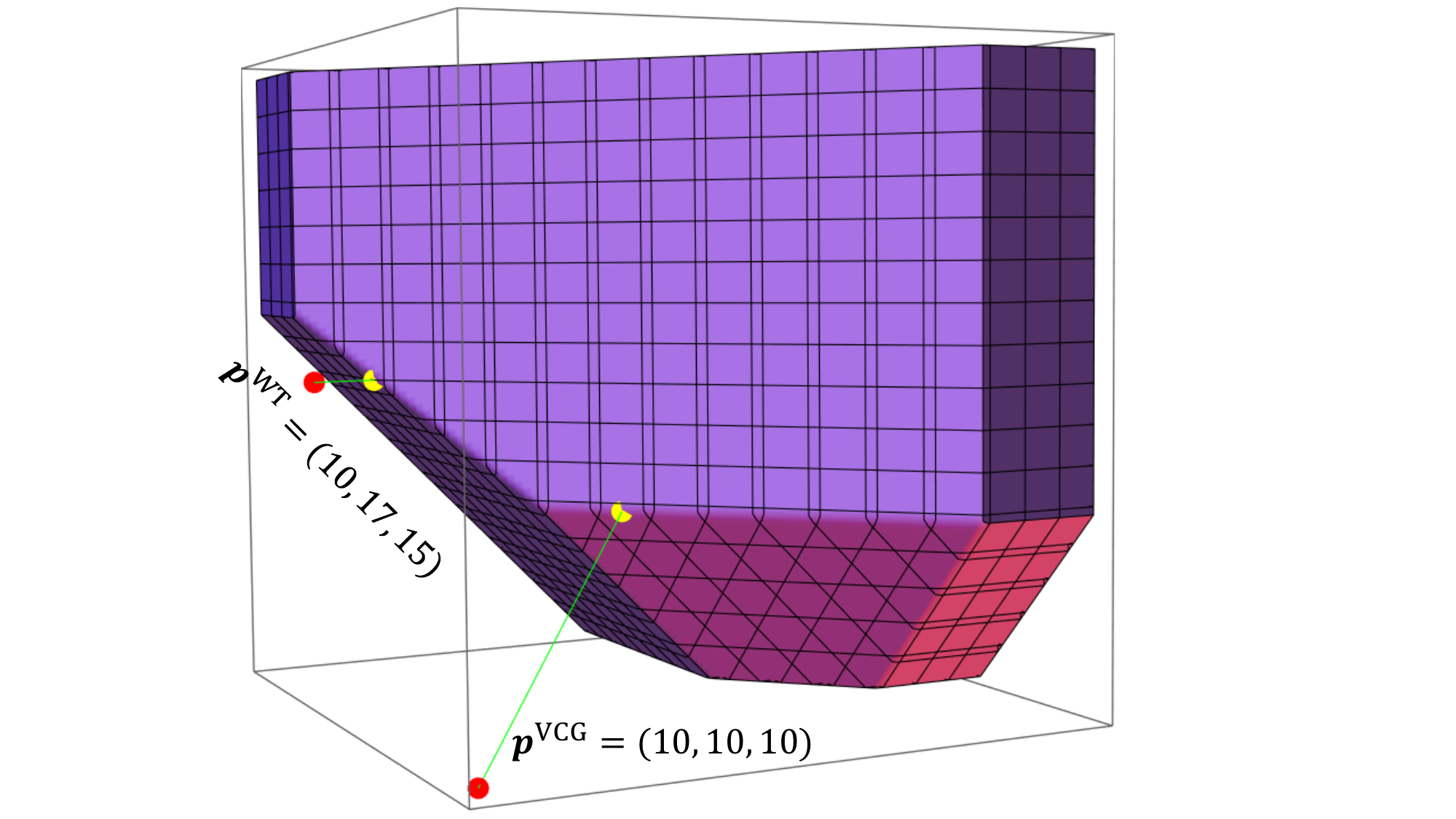}
    \caption{Price vectors $\vec{p}^{\VCG}$ and $\vec{p}^{\WT}$ (in red) and their nearest respective minimum-revenue core points (in yellow, connected by a green line) as derived in Example~\ref{example:core}. $\MRC(\vec{p}^{\WT})$ lies on a different face of the core than $\MRC(\vec{p}^{\VCG})$ and is of higher revenue.}
    \label{fig:core}
\end{figure}

\section{Computing Weakest-Type Prices}\label{sec:formulations}
In this section we develop techniques to compute $\vec{p}^{\WT}$, which are needed as a subroutine for computing the payments of our new core-selecting CAs.~\citet{balcan2023bicriteria} provide an initial theoretical investigation of WT computation, and one of our approaches builds upon their formulation, but we are the first to develop practical techniques and evaluate them via experiments.

Recall $B_i\subseteq 2^M$ is the set of bundles bidder $i$ bids on, so, for each $S\in B_i$, bidder $j$ submits her value $v_i(S)$ which is the maximum amount she would be willing to pay to win bundle $S$. For $\vec{B} = (B_1,\ldots, B_n)$, $\Gamma(\vec{B})$ denotes the set of feasible selections of winning bids.

%{\color{blue} Could have a section about weakest-type formulations for different bidding languages and the resulting WDP formulations. E.g., the language for large item auctions by~\citet{goetzendorff2015compact}, the FUEL bid language for spectrum~\citet{bichler2023taming}, SATS~\citep{weiss2017sats}, tree-based bidding language~\cite{Parkes05:ICE}. Read FUEL paper.}

\subsection{Background on Winner Determination Formulations}
The standard integer programming formulation of winner determination involves variables $x_j(S)$ indicating whether bidder $j$ is allocated bundle $S$:
\begin{equation*}\label{eq:wd}
w(\vec{v}) = \max\left\{\sum_{j\in N}\sum_{S\in B_j}v_j(S) x_j(S) :
\begin{aligned}
    &\sum_{j\in N}\sum_{S\in B_j, S\ni i} x_j(S)\le 1\;\;\forall i\in M \\
    & \sum_{S\in B_j} x_j(S)\le 1 \;\;\forall j\in N \\
    & x_j(S)\in\{0,1\}\;\; \forall j\in N, S\in B_j
\end{aligned}
\right\}.%\tag{WD-1}
\end{equation*} The first set of constraints ensures that no item is over-allocated and the second set of constraints ensures that each bidder has at most one winning bid. This formulation is the de facto method for computing efficient allocations in practice, and is the formulation we use when solving winner determination problems in experiments. 

Winner determination can also be formulated as the following linear program (LP) due to~\citet{Bikhchandani02:Package} (see also~\citet{deVries2007ascending}). The linear program explicitly enumerates all possible feasible allocations of the items and has the property that its optimal solution is integral. It involves variables $x_j(S)$ indicating whether bundle $S$ is allocated to bidder $j$ and variables $\delta(\vec{S})$ indicating whether feasible allocation $\vec{S}\in\Gamma$ is chosen:
\begin{equation*}\label{eq:wd-lp}
\max\left\{\sum_{j\in N}\sum_{S\in B_j}v_j(S) x_j(S) :
\begin{aligned}
    &x_j(S)\le\sum_{\vec{S}\in\Gamma: S_j = S}\delta(\vec{S})\;\;\forall j\in N, S\in B_j\;\; \boxed{p_j(S)}\\
    & \sum_{S\in B_j} x_j(S)\le 1 \;\;\forall j\in N \;\; \boxed{\pi_j}\\
    & \sum_{\vec{S}\in\Gamma}\delta(\vec{S})\le 1\;\; \boxed{\pi_{s}} \\
    & x_j(S)\ge 0\;\; \forall j\in N, S\in B_j \\
    & \delta(\vec{S})\ge 0\;\;\forall\vec{S}\in\Gamma
\end{aligned}
\right\}.%\tag{WD-2}
\end{equation*}
The first set of constraints ensures that winning bids are consistent with the bundles in the efficient allocation, the second set of constraints ensures that each bidder has at most one winning bid, and the third constraint ensures that only one efficient allocation is chosen. The corresponding dual variables are boxed following the respective primal constraint. The dual LP is given by (with the corresponding primal variables boxed) \begin{equation*}\label{eq:wd-lp-dual}
\min\left\{\sum_{j\in N}\pi_j + \pi_s :
\begin{aligned}
    &\pi_j\ge v_j(S)-p_j(S)\;\;\forall j\in N, S\in B_j\;\; \boxed{x_j(S)}\\
    & \pi_s\ge\sum_{j\in N}p_j(S_j)\;\;\forall\vec{S}\in\Gamma\;\;\boxed{\delta(\vec{S})}
\end{aligned}
\right\}%\tag{WD-2-D}
\end{equation*} and has a constraint for every possible feasible allocation $\vec{S}\in\Gamma$. By strong duality, its optimal objective value is also $w(\vec{v})$. The dual variables $p_j(S)$ have the natural interpretation of non-additive non-anonymous bundle prices that support the efficient allocation computed by the primal, with $\pi_j$ representing bidder $j$'s utility and $\pi_s$ the seller's revenue~\citep{Bikhchandani02:Package, deVries2007ascending}, though in general these do not coincide with VCG prices (\citet{Bikhchandani01:Linear} provide an in depth exploration of the connections between LP duality and VCG prices).

\subsection{Formulations and Constraint Generation for WT Computation}

Let $\widetilde{B}_i$ denote the set of bundles $S_i$ such that $v_i(S_i)$ is constrained by $\Theta_i(\vec{v}_{-i})$ (so if $\Theta_i(\vec{v}_{-i})$ is explicitly represented as a list of constraints on $v_i$, $\widetilde{B}_i$ is the set of bundles $S_i$ such that $v_i(S_i)$ appears in one of those constraints). 

We consider two mathematical programming formulations of weakest-competitor VCG price computation, which is the min-max optimization problem $$\min_{\widetilde{v}_i\in\Theta_i(\vec{v}_{-i})}\max_{\vec{S}\in\Gamma(\widetilde{B}_i, \vec{B}_{-i})}\widetilde{v}_i(S_i)+\sum_{j\in N\setminus i}v_j(S_j).$$ The first is due to~\citet{balcan2023bicriteria} and the second is based on the dual of the winner determination LP~\eqref{eq:wd-lp-dual} of~\citet{Bikhchandani02:Package}. Both formulations enumerate the set of feasible allocations $\Gamma$ in their constraint set so they are too large to be written down explicitly. Instead, we develop constraint generation routines that dynamically add constraints as needed.

\subsubsection*{Formulation based on Balcan-Prasad-Sandholm LP}

The mathematical program for computing the pivot term in bidder $i$'s weakest-competitor VCG price $p_i^{\WT}$ due to~\citet{balcan2023bicriteria} is:

\begin{equation}\label{eq:wc-bps}\min\left\{\gamma :
\begin{aligned}
    &\widetilde{v}_i(S_i) + \sum_{j\neq i}v_j(S_j)\le\gamma\;\;\forall \vec{S}\in\Gamma(\widetilde{B}_i, \vec{B}_{-i}), \\
    &\widetilde{v}_i\in \Theta_i(\vec{v}_{-i})
\end{aligned}
\right\}.\tag{BPS}
\end{equation} It turns the min-max problem into a pure minimization problem by enumerating the inner maximization terms and adding an auxiliary scalar variable $\gamma$ to upper bound those terms. In constraint generation, we initialize the BPS program with some restricted set of constraints corresponding to feasible allocations $\Gamma_0\subseteq\Gamma(\widetilde{B}_i,\vec{B}_{-i})$ and solve to get a candidate solution $\widehat{\gamma}, \widehat{v}_i$. Next, we find the most violated constraint not currently in $\Gamma_0$ by computing $w(\widehat{v}_i,\vec{v}_{-i})$ and comparing to $\widehat{\gamma}$. If $\widehat{\gamma} - w(\widehat{v}_i, \vec{v}_{-i}) < 0$ we have found a (most) violated constraint, and we add the constraint corresponding to the violating allocation $(\widehat{S}_1,\ldots, \widehat{S}_n)$ that solves $w(\widehat{v}, \vec{v}_{-i})$ to the restricted pricing LP (that is, $\Gamma_0\leftarrow\Gamma_0\cup\{\widehat{\vec{S}}\}$). The BPS program with the additional constraint is resolved and the process iterates. Otherwise if $\widehat{\gamma} - w(\widehat{v}_i, \vec{v}_{-i}) \ge 0$, all constraints of the unrestricted BPS program are satisfied and so $\widehat{\gamma},\widehat{v}_i$ is an optimal solution to the BPS program and constraint generation terminates. %To initialize the pricing LP without any additional computation, we can use the single constraint corresponding to the efficient allocation $(S_1^*,\ldots, S_n^*)$.

\subsubsection*{Formulation based on Bikhchandani-Ostroy LP}

The mathematical program for computing the pivot term in bidder $i$'s weakest-competitor VCG price $p_i^{\WT}$ based on the dual LP of the~\citet{Bikhchandani02:Package} formulation~\eqref{eq:wd-lp-dual} is:

\begin{equation}\label{eq:wc-bo}
\min\left\{\sum_{j\in N}\pi_j + \pi_s :
\begin{aligned}
    &\pi_i\ge \widetilde{v}_i(S) - p_i(S)\;\;\forall S\in \widetilde{B}_i\\
    &\pi_j\ge v_j(S)-p_j(S)\;\;\forall j\in N\setminus i, S\in B_j\\
    & \pi_s\ge\sum_{j\in N}p_j(S_j)\;\;\forall\vec{S}\in\Gamma(\widetilde{B}_i, \vec{B}_{-i})\\
    &\widetilde{v}_i\in\Theta_i(\vec{v}_{-i})
\end{aligned}
\right\}.\tag{BO}
\end{equation} It has $|\widetilde{B}_i| +\sum_{j\neq i}|B_j| + n + 1$ variables while the BPS formulation has $|\widetilde{B}_i|+1$ variables. In constraint generation, we solve a restricted BO program over an initial set of feasible allocations $\Gamma_0$ (replacing the third set of constraints) and get a candidate solution $\widehat{v}_i$, $(\widehat{\pi}_j)$, $\widehat{\pi}_s$, $(\widehat{p}_j(S))$. To find the most violated constraint we solve winner determination where bidders' bids are given by the values of the supporting prices $\widehat{p}_j(S)$ and compare the value to the seller's revenue $\widehat{\pi}_s$. If $\widehat{\pi}_s - w(\widehat{\vec{p}}) < 0$ the constraint corresponding to the feasible allocation $\widehat{\vec{S}}$ that solves $w(\widehat{\vec{p}})$ is (most) violated. So, we add that constraint to the restricted BO program and iterate. Else if $\widehat{\pi}_s - w(\widehat{\vec{p}}) \ge 0$ all BO constraints are satisfied, and our candidate solution is optimal so constraint generation terminates.

%Both formulations enumerate the set of feasible weakest-competitor allocations $\Gamma(\widetilde{B}_i, \vec{B}_{-i})$ in their constraints and are thus too large to be written down explicitly and fed to a solver. However, we show the {\em separation problem} (equivalently, the implementation of a {\em separation oracle}) for both formulations is equivalent to winner determination. Theoretically, this implies both formulations can be solved via the ellipsoid algorithm~\citep{Grotschel93:Geometric} with polynomially many calls to a winner determination oracle. Practically, this means both formulations are amenable to constraint generation techniques.

\paragraph{Remark} An advantage of the BPS formulation is that it can generally be applied to any multidimensional mechanism design problem, and the constraint generation works as long as one can formulate and solve the separation problem (winner determination) in a tractable way. The BO formulation is specific to combinatorial auctions.

So far we have given an abstract presentation of the BPS and BO mathematical programs that make no reference to the structure of type spaces. The constraint generation routines we described work generally as long as one can solve the associated optimization problems, but the best approach to WT computation might differ based on the type space representation. For example, if $\Theta_i(\vec{v}_{-i})$ is a finite list of types, one would simply solve the winner determination problem for each type in the list and choose the smallest value. In our experiments type spaces are polyhedra described by linear constraints, so the BPS and BO formulations are linear programs for which constraint generation as we have described is a de facto approach. Extending these techniques to other forms of type spaces (convex sets, unions of polyhedra---in order to represent disjunctive statements, {\em etc.}) is a compelling direction for future work. 

\subsection{Comparison of the BPS and BO Formulations}

We implemented constraint generation on both the BPS and BO formulations where the type spaces $\Theta_i(\vec{v}_{-i})$ were generated independently at random for each bidder. Each type space $\Theta_i(\vec{v}_{-i})$ was determined by 8 randomly generated linear constraints (so both the BPS and BO formulations were linear programs) that were consistent with bidder $i$'s actual bids (we defer the specific details of how we generated CA instances and respective bidder type spaces to Section~\ref{sec:experiments}). For both formulations, we initialized the starting set of allocations $\Gamma_0$ with only the efficient allocation $\vec{S}^*$.

Run-times and total number of constraint generation iterations to compute $\vec{p}^{\WT}$ are reported in Table~\ref{table:bps_vs_bo}. 

{
\begin{table}[h]
\centering
\begin{tabular}{c|cccc|cccc|}
\cline{2-9}
 &
  \multicolumn{4}{c|}{\small BPS} &
  \multicolumn{4}{c|}{\small BO} \\ \hline
\multicolumn{1}{|c|}{\scriptsize Goods/Bids} &
  \multicolumn{1}{c|}{\begin{tabular}[c]{@{}c@{}}\small Run-time\\ \small (GM)\end{tabular}} &
  \multicolumn{1}{c|}{\begin{tabular}[c]{@{}c@{}}\small Run-time\\ \small (GSD)\end{tabular}} &
  \multicolumn{1}{c|}{\begin{tabular}[c]{@{}c@{}}\small CG iters\\ \small (GM)\end{tabular}} &
  \begin{tabular}[c]{@{}c@{}}\small CG iters\\ \small (GSD)\end{tabular} &
  \multicolumn{1}{c|}{\begin{tabular}[c]{@{}c@{}}\small Run-time\\ \small (GM)\end{tabular}} &
  \multicolumn{1}{c|}{\begin{tabular}[c]{@{}c@{}}\small Run-time\\ \small (GSD)\end{tabular}} &
  \multicolumn{1}{c|}{\begin{tabular}[c]{@{}c@{}}\small CG iters\\ \small (GM)\end{tabular}} &
  \begin{tabular}[c]{@{}c@{}}\small CG iters\\ \small (GSD)\end{tabular} \\ \hline
\multicolumn{1}{|c|}{64/250} &
  \multicolumn{1}{c|}{4.0} &
  \multicolumn{1}{c|}{2.3} &
  \multicolumn{1}{c|}{25.6} &
  2.4 &
  \multicolumn{1}{c|}{7.3} &
  \multicolumn{1}{c|}{2.3} &
  \multicolumn{1}{c|}{58.2} &
  2.5 \\ \hline
\multicolumn{1}{|c|}{64/500} &
  \multicolumn{1}{c|}{9.4} &
  \multicolumn{1}{c|}{2.9} &
  \multicolumn{1}{c|}{45.8} &
  3.1 &
  \multicolumn{1}{c|}{20.8} &
  \multicolumn{1}{c|}{2.8} &
  \multicolumn{1}{c|}{102.5} &
  2.6 \\ \hline
\multicolumn{1}{|c|}{128/250} &
  \multicolumn{1}{c|}{8.5} &
  \multicolumn{1}{c|}{3.2} &
  \multicolumn{1}{c|}{42.2} &
  2.5 &
  \multicolumn{1}{c|}{17.6} &
  \multicolumn{1}{c|}{2.9} &
  \multicolumn{1}{c|}{110.3} &
  2.4 \\ \hline
\multicolumn{1}{|c|}{128/500} &
  \multicolumn{1}{c|}{46.9} &
  \multicolumn{1}{c|}{6.3} &
  \multicolumn{1}{c|}{59.8} &
  2.9 &
  \multicolumn{1}{c|}{110.6} &
  \multicolumn{1}{c|}{5.8} &
  \multicolumn{1}{c|}{164.2} &
  2.4 \\ \hline
\end{tabular}
\caption{Geometric mean (GM) and standard deviation (GSD) of run-times (in seconds) and number of constraint generation (CG) iterations for the BPS and BO formulations, varying the number of goods and bids, averaged across 100 instances for each good/bid setting.}
\label{table:bps_vs_bo}
\end{table}}

Constraint generation on the BPS formulation was significantly faster and required far fewer iterations than the BO formulation. On an instance-by-instance basis, the BPS formulation was faster and cheaper than the BO formulation on 100\% of CA instances. In all experiments reported in the following section (Section~\ref{sec:experiments}), we therefore only ran constraint generation on the BPS formulation for all WT computations.

\section{Experiments}\label{sec:experiments}

We ran experiments to evaluate the revenue, incentive, fairness, and computational properties of our new core-selecting CAs. We describe the main components of the experimental setup below.

\subsection*{New and old core-selecting CAs}

For a given CA instance we compare five different core-selecting payment rules, defined in the below bulleted list (the three new CAs we introduce in this work are bolded):
\begin{itemize}
    \item Vanilla VCG nearest~\citep{day2012quadratic}: the point $\vec{p}\in\MRC(\vec{p}^{\VCG})$ that minimizes $\norm{\vec{p}-\vec{p}^{\VCG}}_2^2$.
    \item Vanilla zero nearest~\citep{erdil2010new}: the point $\vec{p}\in\MRC(\vec{p}^{\VCG})$ that minimizes $\norm{\vec{p}}_2^2$.
    \item {\bf WT nearest}: the point $\vec{p}\in\MRC(\vec{p}^{\WT})$ that minimizes $\norm{\vec{p}-\vec{p}^{\WT}}_2^2$.
    \item {\bf Zero nearest}: the point $\vec{p}\in\MRC(\vec{p}^{\WT})$ that minimizes $\norm{\vec{p}}_2^2$.
    \item {\bf VCG nearest}: the point $\vec{p}\in\MRC(\vec{p}^{\WT})$ that minimizes $\norm{\vec{p}-\vec{p}^{\VCG}}_2^2$.
\end{itemize}

The WT-nearest rule is the most direct generalization of the vanilla VCG-nearest rule proposed by~\citet{day2012quadratic} and the zero-nearest rule is the most direct generalization of the vanilla zero-nearest rule proposed by~\citet{erdil2010new}.

\subsubsection*{Quadratic programming and core-constraint generation} Each of the five price vectors is computed via the quadratic programming and core-constraint generation technique developed by~\citet{day2012quadratic}, which we describe here at a high level. Details can be found in~\citet{day2012quadratic}. Given an input {\em reference point} $\vec{p}^{\texttt{ref}}$, the goal is to find the point $\vec{p}\in\MRC(\vec{p}^{\WT})$ that minimizes $\norm{\vec{p}-\vec{p}^{\texttt{ref}}}_2^2$. Let $\text{QP}(r)$ denote the quadratic program $$\min\left\{\norm{\vec{p}-\vec{p}^{\texttt{ref}}}_2^2: \vec{p}\in\core, \vec{p}\ge\vec{p}^{\WT}, \norm{\vec{p}}_1=r\right\}$$ and let LP denote the linear program $$\min\left\{\norm{\vec{p}}_1 : \vec{p}\in\core, \vec{p}\ge\vec{p}^{\WT}\right\}.$$ Core constraints make both formulations too large to represent explicitly, and hence they are solved with constraint generation. First, solve restricted LP with some initial set (possibly empty) of core constraints; let $\widehat{r}$ be the optimal objective. Then, solve restricted $\text{QP}(\widehat{r})$ with the same initial core constraints, and let $\widehat{\vec{p}}$ be the optimal solution. To find the most violated core constraint, solve an auxiliary winner determination where all bids by winner $i$ are reduced by their opportunity cost $v_i(S_i^*) - \widehat{p}_i$. If the optimal winner determination value/efficient welfare is more than the current QP revenue $\norm{\widehat{\vec{p}}}_1$, the core constraint corresponding to the set of winners in the auxiliary winner determination is violated. Add that constraint to the restricted LP and QP, resolve LP to get an updated $\widehat{r}$, solve $\text{QP}(\widehat{r})$, and iterate. (The revenue-minimization LP is needed to ensure that we find the closest point to $\vec{p}^{\texttt{ref}}$ on $\MRC(\vec{p}^{\WT})$. Without that we might find a closer point, but it will be outside the minimum-revenue core and therefore not minimize the sum of incentives to deviate.)

\subsection*{Type space generation}

For each CA instance, we generated synthetic bidder type spaces $\Theta_i(\vec{v}_{-i})$ determined by linear constraints (so the formulations for WT price computation from Section~\ref{sec:formulations} are LPs). We generated $\Theta_i(\vec{v}_{-i})$ independently for each bidder by generating $K$ random linear constraints according to parameter $\beta$ as follows. Each constraint is of the form $$\sum_{S_i\in B_i}X(S_i)c(S_i)\widetilde{v}_i(S_i) \ge \alpha\cdot\sum_{S_i\in B_i}X(S_i)c(S_i)v_i(S_i)$$ where $\widetilde{v}_i(S_i)$ are the variables representing bidder $i$'s bids, each $X(S_i)$ is an independent Bernoulli $0/1$ random variable with success probability $\beta$, each $c(S_i)$ is drawn uniformly and independently from a decay distribution where $c(S_i)$ is initially equal to $1$ and is repeatedly incremented with success probability $0.2$ until failure, and $\alpha$ is drawn uniformly at random from $[1/2, 1]$. So, each such constraint is guaranteed to be satisfied by the actual bids, and $\alpha$ determines how close to tight the constraint is. Each of the $K$ constraints per-bidder is generated this way independently.

\subsection{Results}

We used the Combinatorial Auction Test Suite (CATS)~\citep{Leyton-Brown00:Towards} version 2.1 to generate CA instances. Like~\citet{Day07:Fair} and~\citet{day2012quadratic}, we generated each instance from a randomly chosen distribution from the seven available distributions meant to model real-world CA applications. Code for our experiments was written in C{\tt ++} and we used Gurobi version 12.0.1, limited to 8 threads, to solve all linear programs, integer programs, and quadratic programs. All computations were done on a 64-core machine with 16GB of RAM allocated for each CA instance.

\subsubsection*{Run-time cost of WT computation}

Table~\ref{table:beta_runtime} records the effects of varying $\beta\in\{0.2, 0.5, 0.8\}$ (which controls the sparsity of type space constraints) on the run-time and number of CG iterations to compute $\vec{p}^{\WT}$. We fixed the number of constraints $K = 8$, and for each $\beta$ and each setting of goods in $\{64,128\}$ and bids in $\{250,500\}$ generated 100 instances (for a total of 400 instances). For these instances, the (geometric) mean run-time and worst-case run-time for $\vec{p}^{\VCG}$ were 2.7 seconds and 608.5 seconds, respectively.

\begin{table}[h]
\centering
\begin{tabular}{|c|c|c|c|c|}
\hline
$\beta$ &
  \begin{tabular}[c]{@{}c@{}}Run-time\\ (GM)\end{tabular} &
  \begin{tabular}[c]{@{}c@{}}CG iters\\ (GM)\end{tabular} &
  \begin{tabular}[c]{@{}c@{}}Run-time\\ (Max)\end{tabular} &
  \begin{tabular}[c]{@{}c@{}}CG iters\\ (Max)\end{tabular} \\ \hline
0.2 & 9.9  & 32.0 & 1515.1 & 424 \\ \hline
0.5 & 13.2 & 56.4 & 1610.0 & 536 \\ \hline
0.8 & 15.3 & 75.3 & 1896.0 & 567 \\ \hline
\end{tabular}
\caption{Run-times and constraint generation iterations for the BPS formulation as $\beta$ varies, with number of goods varying in $\{64,128\}$ and number of bids varying in $\{250,500\}$, averaged over 100 instances for each $\beta$ and each setting of goods/bids.}
\label{table:beta_runtime}
\end{table}

The worst run-time for WT computation was thus roughly $3.1\times$ the worst run-time for VCG computation. In general, increasing $\beta$, which increases the density of the type space constraints, increases the cost of WT computation. The additional run-time cost for finding a $\MRC(\vec{p}^{\WT})$ via core-constraint generation was in fact less expensive than the run-time of core-constraint generation to find vanilla $\MRC$ points. The geometric mean runtime of the vanilla VCG nearest rule of~\citet{day2012quadratic} on the above instances was 1.7 seconds, with a worst case run-time of 523.9 seconds. The geometric mean of our WT nearest rule on the same instances was 1.0 seconds, with a worst case run-time of 475.0 seconds. So, the main run-time cost of our new core-selecting CAs is in computing $\vec{p}^{\WT}$.

Varying the number of type space constraints $K$ did not have a significant impact on run-time nor number of constraint generation iterations for WT computation. Over all CA instances with number of goods in $\{64, 128\}$ number of bids in $\{250, 500, 1000\}$, and $\beta = 0.3$, the geometric mean of run-times over all $K$ was 19.7 seconds and the geometric mean of constraint generation iterations was 42.7. The worst-case VCG run-time was 19545.2	seconds and the worst-case WT run-time was 47718.0	seconds ($2.44\times$ larger than VCG run-time). The significantly larger run-time relative to the experiment varying $\beta$ was due to the inclusion of the the CATS instances with 1000 bids.

\subsubsection*{Incentive and revenue effects}

We now discuss the impact of type space information on the sum of bidders' incentives to deviate from truthful bidding in a $\MRC(\vec{p}^{\WT})$-selecting CA. That is, we record the quantity $\sum_{i\in 
W}\delta_i^{\vec{p}}$ where $\vec{p}$ is any one of our new core-selecting CAs. By Theorem~\ref{theorem:best_response} this is equal to $\sum_{i\in W} p_i - p_i^{\WT}$, that is, the difference in revenue between the $\MRC(\vec{p}^{\WT})$-selecting rule and WT. We track this quantity as the number of type space constraints $K$ per bidder varies in $\{1,2,4,8,16\}$, and compare it to the sum of bidders' incentives to deviate in the vanilla unrestricted setting, which by~\citet{day2010optimal} is equal to the difference in revenue between a $\MRC(\vec{p}^{\VCG})$-selecting rule and VCG. Each revenue difference recorded on the $y$-axis of Figure~\ref{fig:incentives_plot} is averaged over 100 CA instances each for goods in $\{64,128\}$ and bids in $\{240,500,1000\}$, for a total of $600$ CA instances and a total of $600\times 5 = 3000$ type space instances/WT computations. We fixed the constraint sparsity parameter $\beta = 0.3$. Figure~\ref{fig:incentives_plot} shows a clear trend that more information about the bidders (in the form of more type space constraints) yields better core incentives---and vastly better incentives than a vanilla MRC-selecting rule.

\begin{figure}[t]
    \centering
    \includegraphics[width=\linewidth]{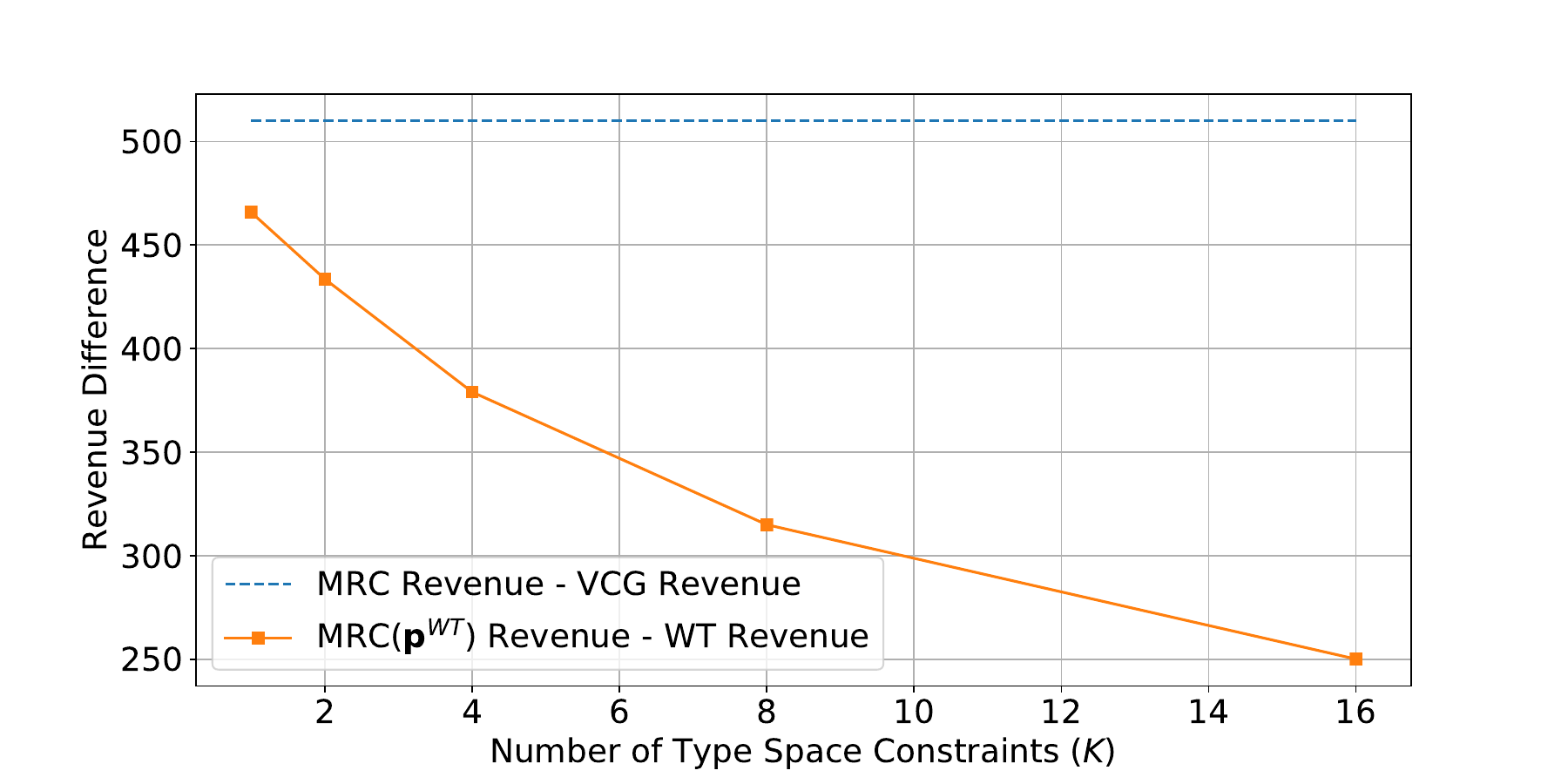}
    \caption{Incentive effects as type spaces convey more information (by varying the number of constraints $K\in\{1,2,4,8,16\}$, with number of goods varying in $\{64,128\}$ and number of bids varying in $\{250,500,1000\}$, averaged over 100 instances for each $K$ and each setting of goods/bids.}
    \label{fig:incentives_plot}
\end{figure}

On the revenue front, Figure~\ref{fig:revenue_plot} shows the impact of more informative type spaces on the revenues generated by our new core-selecting CAs (the experimental setup is the same as in the previous paragraph). The $\MRC(\vec{p}^{\WT})$-selecting rules are the clear winner, nearly closing half the gap between MRC revenue and the efficient social welfare when type spaces are determined by $K=16$ constraints. While the $\MRC(\vec{p}^{\WT})$ revenue is not significantly larger than the MRC revenue for $K\le 8$, WT's revenue {\em is} much larger than VCG's, leading to much better incentives for the $\MRC(\vec{p}^{\WT})$ rule than the MRC rule in that regime despite similar revenues. So, a $\MRC(\vec{p}^{\WT})$-selecting rule with revenue not much larger than a vanilla $\MRC(\vec{p}^{\VCG})$-selecting rule can still provide significantly better incentives for bidders if $\norm{\vec{p}^{\WT}}_1$ is much larger than $\norm{\vec{p}^{\VCG}}_1$. 

\begin{figure}[t]
    \centering
    \includegraphics[width=\linewidth]{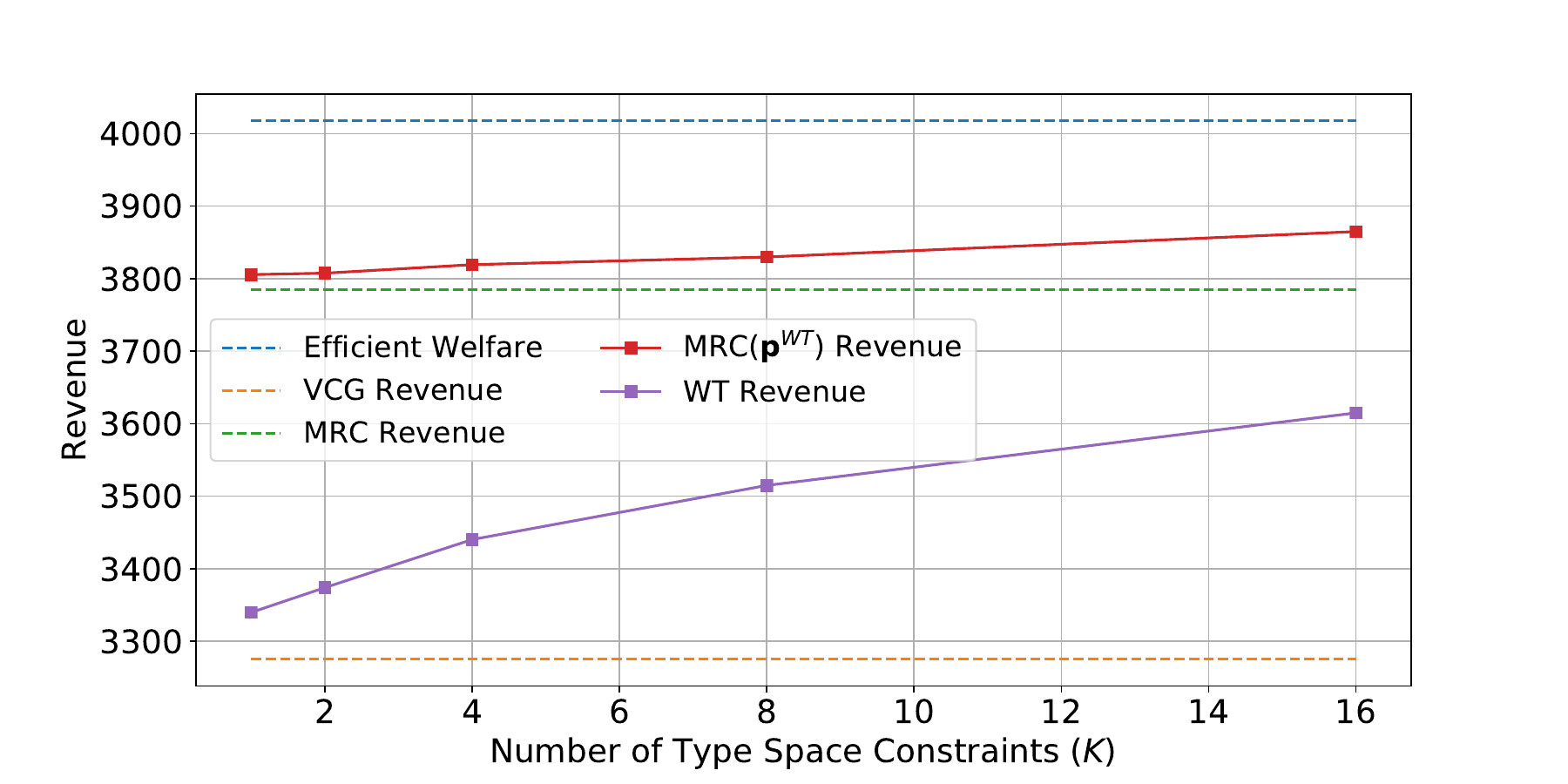}
    \caption{Revenue effects as type spaces convey more information (by varying the number of constraints $K\in\{1,2,4,8,16\}$, with number of goods varying in $\{64,128\}$ and number of bids varying in $\{250,500,1000\}$, averaged over 100 instances for each $K$ and each setting of goods/bids.}
    \label{fig:revenue_plot}
\end{figure}

\subsubsection*{How often is WT in the core?} We now report on the frequency with which $\vec{p}^{\WT}\in\core$. For CA instances with this property, all $\MRC(\vec{p}^{\WT})$-selecting rules return $\vec{p}^{\WT}$ unmodified. The following table records the frequency as the number of type space constraints $K$ varies (the setup is the same as those in Figures~\ref{fig:incentives_plot} and~\ref{fig:revenue_plot}).

\begin{table}[h]
\centering
\begin{tabular}{|l|l|l|l|l|l|}
\hline
\small Number of type space constraints &
  1 &
  2 &
  4 &
  8 &
  16 \\ \hline
\begin{tabular}[c]{@{}l@{}}\small Fraction of CAs where $\vec{p}^{\VCG}\notin\core$ but $\vec{p}^{\WT}\in\core$\end{tabular} &
  2.00\% &
  3.02\% &
  3.36\% &
  3.69\% &
  4.18\% \\ \hline
\end{tabular}
\label{table:core_fraction}
\caption{Frequency with which WT is in the core but VCG is not, with number of goods varying in $\{64,128\}$ and number of bids varying in $\{250,500,1000\}$; 100 instances for each $K$ and each setting of goods/bids.}
\end{table}

As the number of type space constraints increases, the likelihood that WT is in the core increases as well (which is in accordance with the trend in Figure~\ref{fig:revenue_plot} that more information through more type space constraints implies greater WT revenue). Additionally, 1.17\% of all instances had the property that both VCG and WT were in the core, both generating nonzero revenue (this has no dependence on the type space since if VCG is in the core and generates nonzero revenue, so does WT). A fascinating phenomenon we observed was that 7.5\% of instances had the property that all vanilla MRC-selecting rules (like vanilla VCG-nearest of~\citet{day2012quadratic} and vanilla zero-nearest of~\citet{erdil2010new}) generated zero revenue. In other words, VCG generates zero revenue {\em yet is in the core}. This is an {\em even worse} situation than the zero revenue cases described by~\citet{ausubel2017market} and~\citet{ausubel2023vcg} that a vanilla core-selecting rule is unable to fix. The WT auction is therefore indispensable to generate any revenue in these cases. \footnote{In the context of Theorem~\ref{theorem:impossibility}, such situations arise when the core polytope is the box with diagonally opposite points given by the origin (which is equal to VCG) and the winning bid vector. The WT point is strictly in the interior of this box, and the infinitely many points on line segment connecting the origin to the WT point are attainable with an IC auction.} To our knowledge, no prior work discusses this phenomenon.

\subsubsection*{Who shoulders the core burden?}

In~\citet{day2012quadratic}, the impact of core pricing on the highest and lowest bidder is visualized. They show that on CATS instances with few bids (100 or less), their vanilla VCG nearest rule provides a more equitable apportionment of the core burden than the vanilla zero nearest rule of~\citet{erdil2010new}. That trend is less pronounced for the numbers of bids that we consider (250, 500, and 1000), and hence we present a slightly different visualization of the splitting of the core burden. 

For each CA instance, and each core-pricing rule $\vec{p}$, the core burden relative to WT (resp. VCG) of bidder $i$ is the quantity $\frac{p_i - p_i^{\WT}}{\sum_{i\in W}p_i-p_i^{\WT}}$ (resp. $\frac{p_i - p_i^{\VCG}}{\sum_{i\in W}p_i-p_i^{\VCG}}$). We sorted the bidders in ascending order of winning bid $v_i(S_i^*)$, and summed up the total core burdens for the lower and higher halves of bidders. Figure~\ref{fig:core-burden} displays the splitting of core burdens between the lower and higher halves, averaged across all instances with $K=8$. For VCG nearest, WT nearest, and zero nearest, the left bar displays core burdens relative to WT, with the solid black bottom representing the lower half of bidders and the gray top representing the upper half. The right bar displays core burdens relative to VCG, with the darker gray bottom representing the lower half of bidders and the lighter gray top representing the upper half. Only the core burden relative to VCG is displayed for the vanilla VCG nearest~\citep{day2012quadratic} and vanilla zero nearest~\citep{erdil2010new} since it would not make sense to compute core burdens relative to WT for these rules. Overall, there was not a significant difference between WT-nearest, zero-nearest, and VCG-nearest (and this was also the case in~\citet{day2012quadratic} in their comparison of vanilla VCG-nearest and vanilla zero-nearest on CATS instances with more than 250 bids). VCG-nearest placed the least core burden and zero-nearest placed the greatest core burden on the lower half of bidders, and all three rules are similar to the vanilla MRC rules in terms of core burdens relative to VCG. This fact provides further validation for our $\MRC(\vec{p}^{\WT})$-selecting rules as they do not unfairly skew the apportionment of the core burden.

\begin{figure}
    \centering
    \includegraphics[width=\linewidth]{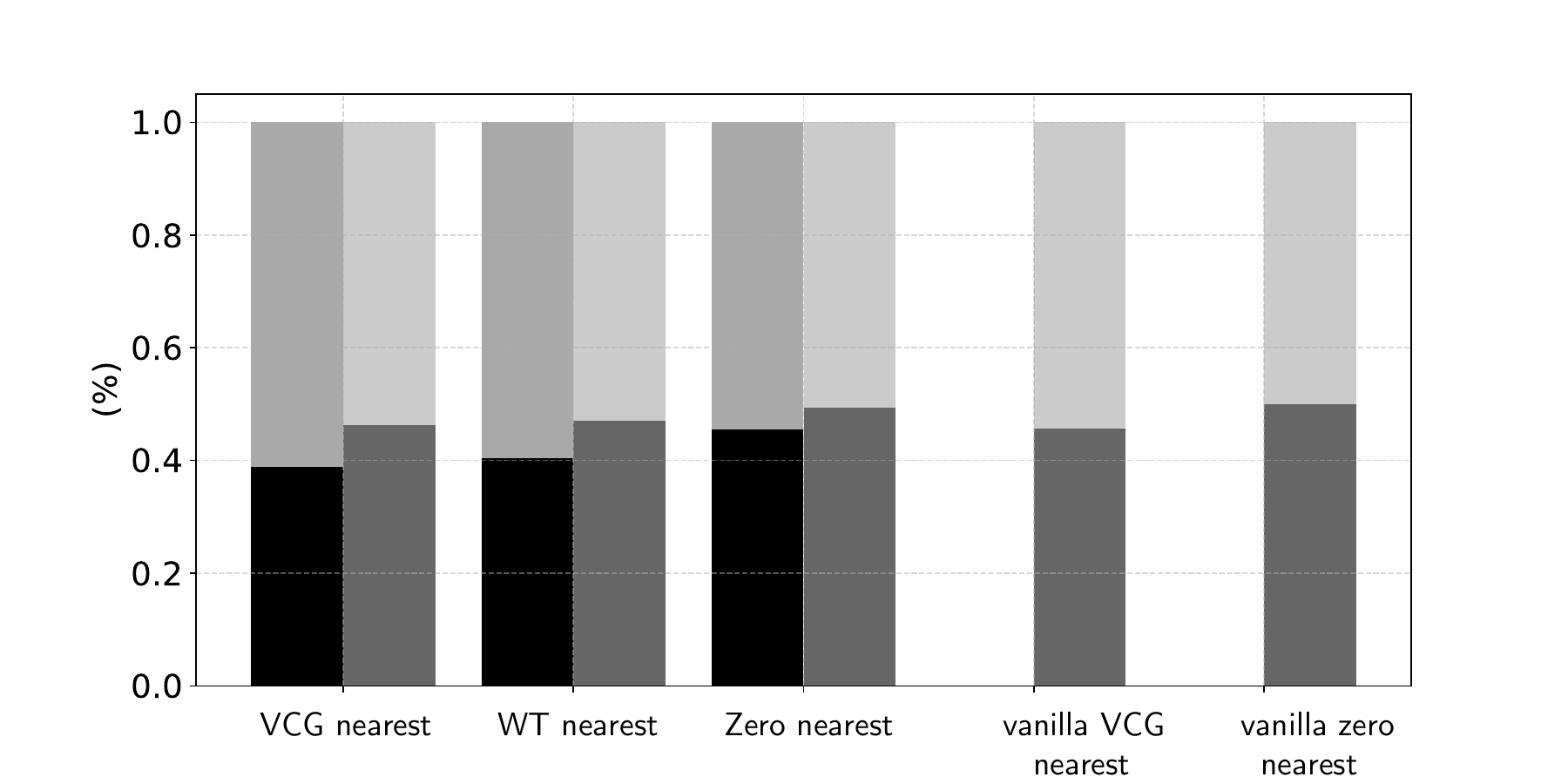}
    \caption{Core burdens shouldered by the lower and upper halves of bidders (measured by winning bid value). For the three $\MRC(\vec{p}^{\WT})$-selecting rules, the left bar displays the core burden split relative to WT, and the right bar displays the core burden split relative to VCG. For the two vanilla MRC-selecting rules, the bar displays the core burden split relative to VCG.}
    \label{fig:core-burden}
\end{figure}

\subsection{Discussion of Alternate Rules}\label{sec:alternate_rules}

We conclude this section with a brief discussion of alternate core-selecting CAs that do not conform to the exact template that has been prescribed here and by prior work.

The previous discussion of equitable sharing of the core burden begets the question of whether there exist core-selecting rules that explicitly enforce how the core burden should be split. For example, is there a $\MRC(\vec{p}^{\WT})$-selecting CA $\vec{p}$ that enforces that each bidder pays a core burden in exact proportion to their winning bid, that is, $$\frac{p_i - p_i^{\WT}}{\sum_{i\in W}p_i - p_i^{\WT}}\ge\alpha\cdot\frac{v_i(S_i^*)}{\sum_{i\in W}v_i(S_i^*)}$$ for some $\alpha$? The answer is no due to the asymmetric information that can be conveyed about bidders by type spaces. For example, if $\Theta_1(\vec{v}_{-1}) = \{v_1\}$, so $p_i^{\WT} = v_i(S_i^*)$, IR constraints force $p_i = p_i^{\WT}$ for any core-selecting $\vec{p}$. So in such situations it could very well be the case that a low bidder is forced to shoulder a large majority of the core burden. A general rule of thumb here appears to be that the bidders with type spaces that convey the least information about them must pay most of the core burden. A formal investigation of this idea is an interesting direction for future research.

Schemes that do not select points on $\MRC(\vec{p}^{\WT})$ are also possible. For example, one could minimize the maximum payment over the subset of the core that lies above $\vec{p}^{\WT}$. This would minimize the worst incentive of any bidder to deviate from truthful bidding (by the same argument as in Theorem~\ref{theorem:general_uc}) rather than the sum of bidders' incentives to deviate, and would still yield a core point that is incentive optimal in a Pareto sense.~\citet{Parkes01:Achieving} proposed such rules in the context of budget balance in exchanges, but those rules have received limited empirical evaluation in the CA setting to date.

Finally, rules that minimize a weighted sum of squares (as studied by~\citet{bunz2022designing}) might be of particular relevance so that bidders $i$ such that $p_i^{\WT}$ is much larger than $p_i^{\VCG}$ are made to pay less of the price difference in moving from $p_i^{\WT}$ to $\MRC(\vec{p}^{\WT})$. 

\section{Conclusions and Future Research}

We presented a new family of core-selecting CAs that take advantage of bidder information known to the auction designer through bidders' type spaces. Our design built upon the WT auction, which boosts revenues beyond VCG by considering for each bidder the weakest type consistent with the auction designer's knowledge. We showed that sufficiently informative type spaces can overcome the well-known impossibility of core-selecting CAs, and gave a revised and generalized impossibility result that depends on whether or not the WT auction is in the core. We then showed that our new family of core-selecting CAs, defined by minimizing revenue on the section of the core above WT prices, minimizes the sum of bidders' incentives to deviate from truthful bidding. This result generalizes those of~\citet{Day07:Fair} and~\citet{day2010optimal} which rely on unrestricted bidder type spaces. On the computational front, we developed new constraint generation techniques for computing WT prices. We compared two formulations, one due to~\citet{balcan2023bicriteria} and a new one based on~\citet{Bikhchandani02:Package} that is a contribution of this paper. Finally, we evaluated our new core-selecting CAs on CATS instances, with synthetic generators for type space constraints. The revenue and incentive benefits of our new CAs, along with their manageable computational overhead, make them a useful addition to the auction design toolkit.

We conclude by discussing avenues for future research. Perhaps the most pressing direction is the development of realistic type space generators by incorporating the specific details of the application domain. Our new CAs display promise on our synthetically-generated type spaces, but to understand their viability in real-world auctions one must develop detailed models of auctioneer knowledge. Generalizing our techniques to type spaces that are not cleanly described by linear constraints is a prerequisite here, as well as accommodating other knowledge structures~\citep{balcan2025increasing}.

A more thorough investigation is needed for the design of $\MRC(\vec{p}^{\WT})$-selecting rules. We introduced three specific ones in this paper (WT nearest, zero nearest, and VCG nearest) that are natural generalizations of vanilla MRC-selecting rules, but as discussed in Section~\ref{sec:alternate_rules} there might be other more economically meaningful rules. A computational study extending~\citet{bunz2022designing} to $\MRC(\vec{p}^{\WT})$-selecting rules is relevant here as well. A promising direction along this vein is to use machine learning to design the reference point, weights, and amplifications of the parameterized rules in~\citet{bunz2022designing}. Explicit equilibrium analysis in the style of~\citet{goeree2016impossibility} and~\citet{ausubel2020core} is important as well.

Our formulations of WT computation were specific to the XOR bidding language. Extensions and modifications of our techniques are needed for other domains and other bidding languages such as those proposed for spectrum auctions~\citep{bichler2023taming,weiss2017sats}, sourcing auctions~\citep{Sandholm02:eMediator,Sandholm13:Very}, and more general domain-independent use~[\citealp{Sandholm99:Algorithm}, \citealp{Nisan00:Bidding}, \citealp{Boutilier01:Bidding}, \citealp{Boutilier02:Solving}]. The interaction between the bidding language of a CA and the language used to express type space knowledge is an unexplored area here as well.

Finally, an important direction within the research strand of {\em mechanism design with predictions}~\citep{balcan2023bicriteria, balkanski2024mechanism} is to relax the assumption that $\vec{v}\in\bTheta$, that is, that type spaces convey {\em correct} information about bidders. How can core-selecting CAs with strong incentive properties be designed using the techniques developed in this paper when type spaces can have small errors? The techniques developed in~\citep{balcan2023bicriteria} in the general setting of revenue-maximizing multidimensional mechanism design will likely be useful here, and can also help shed light on better core selection in mechanism design settings beyond combinatorial auctions.

\subsection*{Acknowledgements}

This material is based on work supported by the NSF under grants IIS-1901403, CCF-1733556, and RI-2312342, the ARO under award W911NF2210266, the Vannevar Bush Faculty Fellowship ONR N00014-23-1-2876, and NIH award A240108S001.

\newpage

\bibliography{references, dairefs}
\bibliographystyle{plainnat}

\end{document}